\documentclass{article}[11pt]
\usepackage{amssymb,latexsym,amsfonts,amsmath,amsthm,amsbsy,amsrefs,subfig, multirow, url, multicol}
\usepackage{tikz}
\usetikzlibrary{snakes}

\usepackage{colortbl}
\definecolor{red}{rgb}{1, 0, 0}
\newcolumntype{r}{>{\columncolor{red}}c}

\newcommand{\disjointUnion}{\sqcup}

\newcommand{\RR}{\ensuremath{\mathbb{R}}}

\newcommand{\bp}{\ensuremath{{\bf p}}}
\newcommand{\bq}{\ensuremath{{\bf q}}}
\newcommand{\bv}{\ensuremath{{\bf v}}}

\newcommand{\ba}{\ensuremath{{\bf a}}}
\newcommand{\bb}{\ensuremath{{\bf b}}}
\newcommand{\bl}{\ensuremath{L}}
\newcommand{\bc}{\ensuremath{{\bf c}}}
\newcommand{\bs}{\ensuremath{{\bf s}}}
\newcommand{\bx}{\ensuremath{{\bf x}}}

\newcommand{\bo}{\ensuremath{{\boldsymbol{\omega}}}}

\newcommand{\smallcdots}{\cdot\hspace{-1mm}\cdot\hspace{-1mm}\cdot}
\newcommand{\rmfill}{\hspace{-1mm}\smallcdots {\bf 0} \hspace{-1mm}\cdot\hspace{-1mm}\cdot\hspace{-.2mm}\cdot\hspace{-.8mm}}

\newlength{\LL}

\newtheorem{theorem}{Theorem}   
 
\newtheorem{rmk}[theorem]{Remark}
\newtheorem{defn}[theorem]{Definition}
\newtheorem{prop}[theorem]{Proposition}

\begin{document}
		

		\title{Combinatorics and the Rigidity of CAD Systems\footnote{Final version to appear in Symposium on Solid and Physical
		 	Modeling '12 and associated special issue of Computer Aided Design.}}

	\author{Audrey Lee-St.John \\
	 	Department of Computer Science \\
		Mount Holyoke College \\
		South Hadley, MA 01075 \\
		\tt{astjohn@mtholyoke.edu}
		\and 
	Jessica Sidman \\
		Department of Mathematics and Statistics\\
		Mount Holyoke College\\
		South Hadley, MA 01075\\
		\tt{jsidman@mtholyoke.edu}}
		\maketitle
	
\begin{abstract}
We study the rigidity of \emph{body-and-cad frameworks} which 
capture the majority of the geometric constraints used in 3D mechanical engineering CAD software.  We present a combinatorial characterization of the generic minimal rigidity of a subset of body-and-cad frameworks in which we treat 20 of the 21 body-and-cad constraints, omitting only point-point coincidences. While the handful of classical combinatorial characterizations of rigidity focus on distance constraints between points, this is the first result simultaneously addressing coincidence, angular, and distance constraints. Our result is 
stated in terms of the partitioning of a graph into edge-disjoint spanning trees. This combinatorial approach provides the theoretical basis for the development of deterministic algorithms (that will not depend on numerical methods) for analyzing the rigidity of body-and-cad frameworks.
\end{abstract}

\section{Introduction}
We study {\em body-and-cad} frameworks composed of rigid bodies with pairwise {\bf c}oincidence, {\bf a}ngular and {\bf d}istance constraints placed between geometric elements (points, lines, planes) rigidly affixed to each body. A framework is {\em flexible} if the bodies can move relative to each other while respecting the constraints;
otherwise, it is {\em rigid} (see Figures \ref{fig.minRigBodyCad} and \ref{fig.flexBodyCadExample}). A fundamental problem for both the user and the CAD software is to determine whether a system is flexible or rigid. 

\begin{figure}[thb]
	\centering \subfloat[A {\bf minimally rigid} body-and-cad framework with (i) a line-line coincidence along the dotted blue axis, (ii) a point-plane coincidence between the red point $r$ and shaded plane, and (iii) a point-point distance between the green points $p$ and $q$.] {\label{fig.minRigBodyCad}
	\begin{minipage}[b]{.46\linewidth}
	\centering
	\includegraphics[width=.7\linewidth]{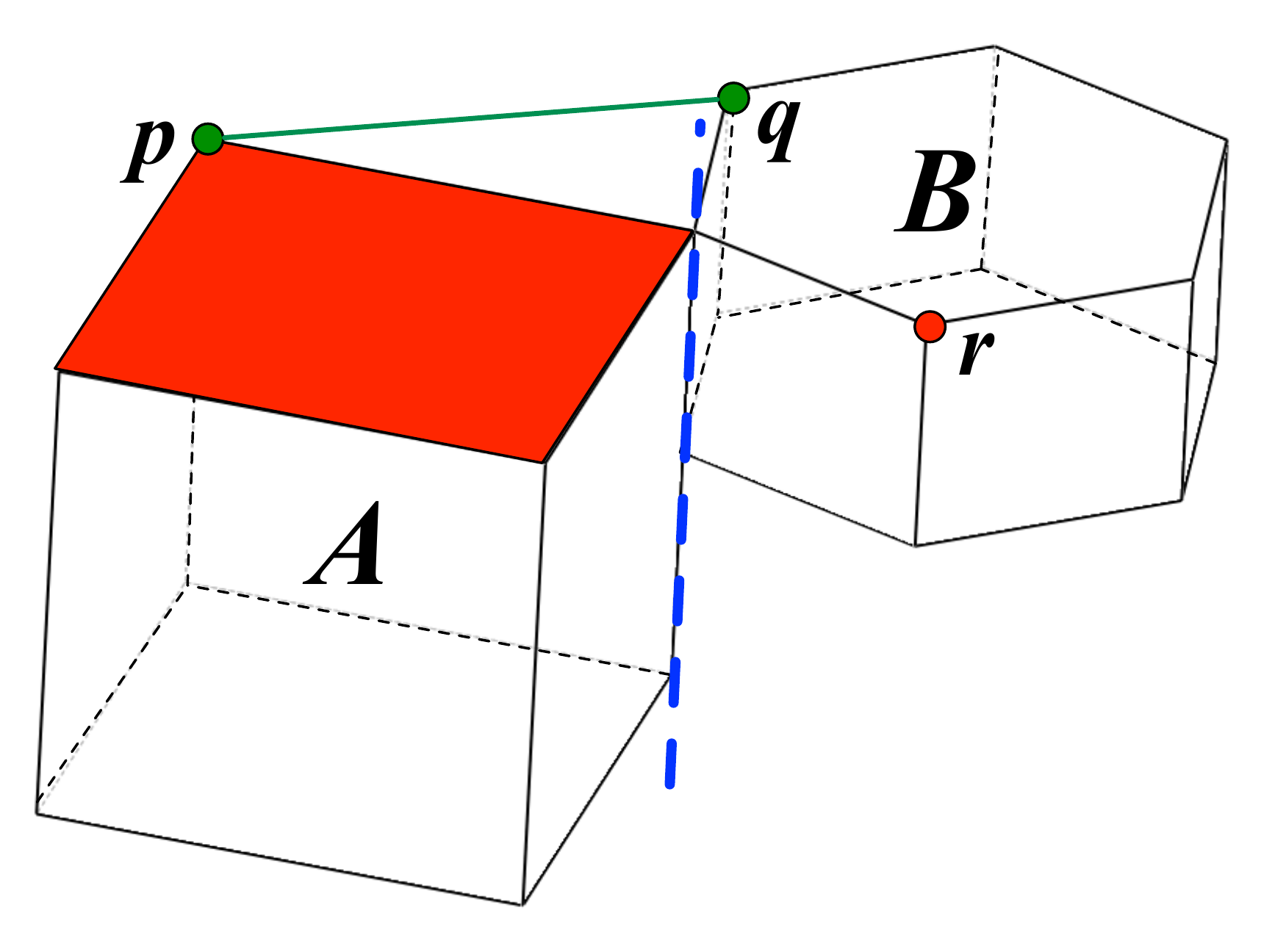}
	\end{minipage}}%
	\hspace{4mm}	
	\centering \subfloat[Removing the point-point distance constraint (iii) results in a {\bf flexible} framework with 1 degree of freedom (rotation about the shared axis).] {\label{fig.flexBodyCadExample}
	\begin{minipage}[b]{.46\linewidth}
	\centering
	\includegraphics[width=.7\linewidth]{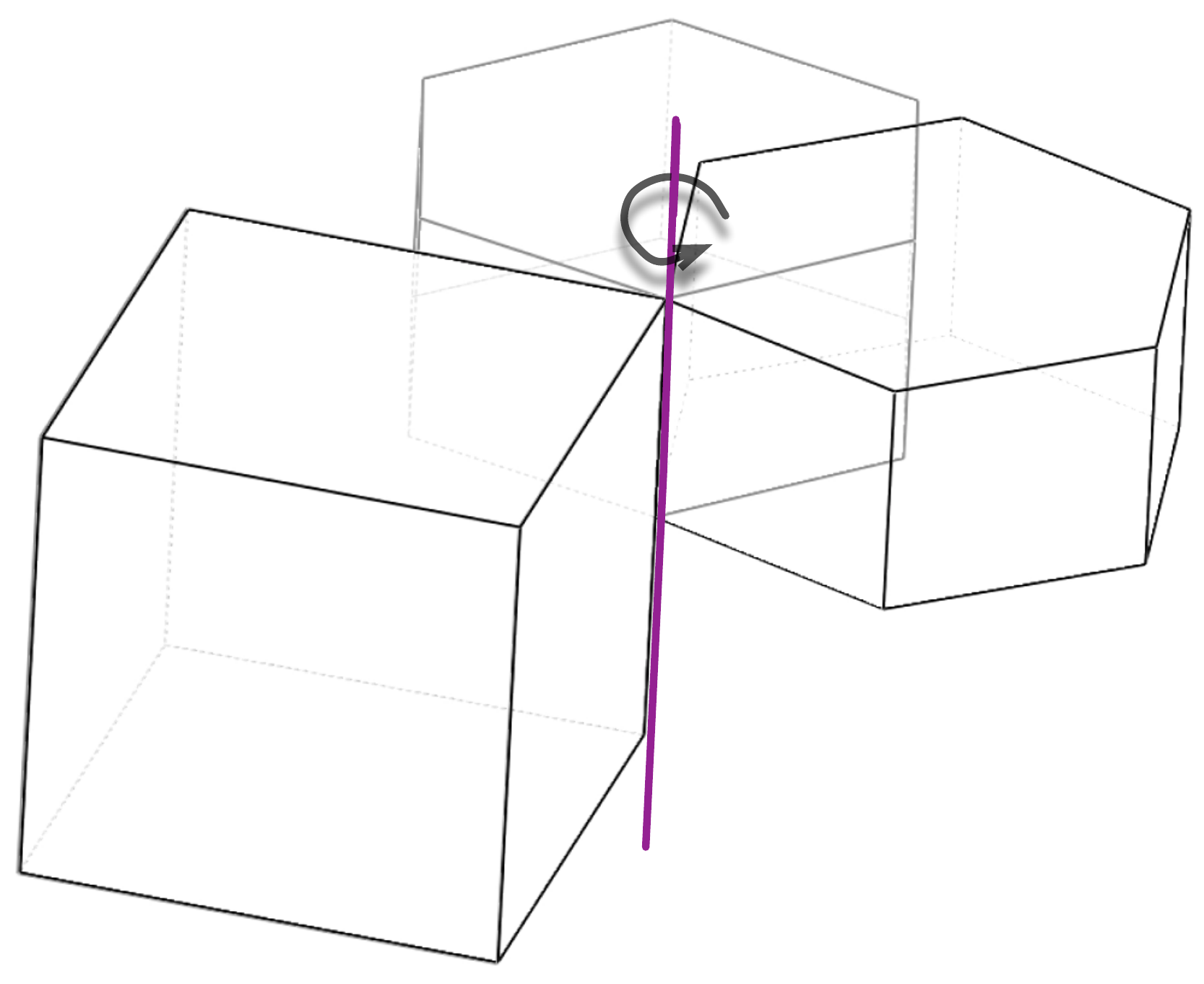}
	\end{minipage}}%
				\caption{A 3D body-and-cad framework.}
				\label{fig.exampleBodyCad}

\end{figure}

\medskip\noindent{\bf Contributions.} 
In this paper, we present a combinatorial characterization of the generic minimal rigidity of a subset of body-and-cad frameworks: we treat 20 of the 21 constraints, omitting point-point coincidences.   The equations governing point-point coincidences exhibit special algebraic behavior (see Appendix \ref{apx:ptPtCoinc}). 

 Body-and-cad frameworks on $n$ bodies without point-point coincidences can be modeled by a graph 
called a $(k,g)$-frame. 
Our main theorem states that generic minimal rigidity of a $(k,g)$-frame is equivalent to a partitioning of the graph's edge set such that each partition consists of the edge-disjoint union of 3 spanning trees.  Combinatorial characterizations of rigidity typically lead to quadratic time algorithms, and we expect this result to have similar implications for the development of efficient algorithms.

\medskip\noindent{\bf Significance.} 
Among the many classes of structures whose rigidity is studied, there are only a handful of cases for which combinatorial characterizations have been proven, with Laman's 2D bar-and-joint \cite{laman:rigidity:1970} and Tay's body-and-bar \cite{tayRigidity} results the most well-known. Furthermore, classical work has focused mainly on distance constraints. To the best of our knowledge, 
{\em ours is the first combinatorial rigidity characterization that simultaneously treats angular and other
constraints}.

\medskip\noindent{\bf Motivation.}
In the 3D ``assembly'' environment of the popular CAD software SolidWorks \cite{solidworks}, 
users place ``mates'' among parts (rigid bodies). The mates specify constraints (e.g., 
distance, angular, coincidences) among {\em geometric elements} (e.g.,
points, lines, planes) identified on the parts. Informative feedback is important as explicit visualization of 
mates is difficult, and mechanical engineers are usually not experts in reasoning about
the underlying geometric constraint systems. 

One of the most difficult situations
arises when a new mate is specified that causes an 
inconsistency in the design. As depicted in Figure \ref{fig.swover},
SolidWorks gives feedback in the form of an alert window 
and highlights a list of previous mates ``overdefining'' the system. This list is intended to 
assist the user in resolving the inconsistency,
but is often too large (sometimes, the entire set of 
mates) to be helprooful.
Current software {\em cannot reliably detect a minimally dependent set of constraints}. 
Indeed, the system in Figure \ref{fig.swover} reveals the severity of software limitations, as
the lift is essentially 2-dimensional. It is composed of two ``scissor'' parts attached
by struts connecting pairs of matching joints. Since the ``scissor'' parts must move identically,
analysis can be reduced to considering a single ``scissor,'' which itself is a planar mechanism.
The development of a complete {\bf rigidity theory} would provide
a rigorous mathematical understanding of these geometric systems.
\begin{figure}[tbh]
	\centering\includegraphics[width=\linewidth]{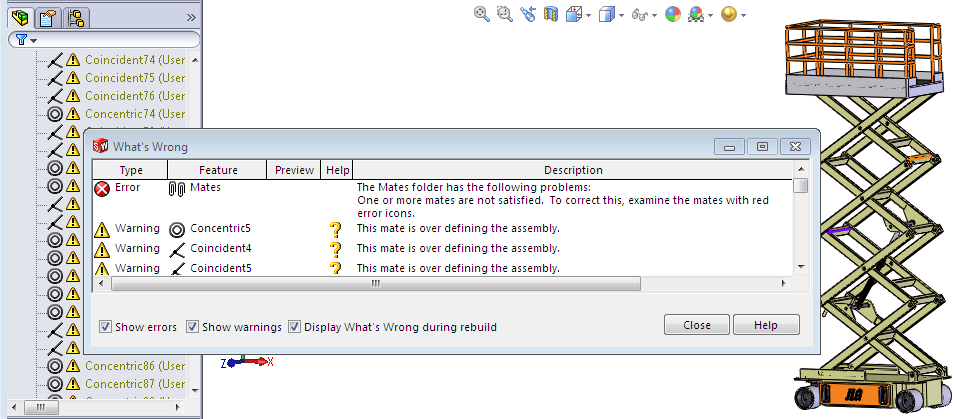}
		\caption{When a new mate is added that results in an inconsistency, previous mates
		``over defining the assembly" are highlighted.  However, the highlighted set is not
		necessarily minimal.   {\em SolidWorks model from {\tt http://www.3dcontentcentral.com}.}}
\label{fig.swover}
\end{figure}

\smallskip\noindent{\bf Potential for algorithmic impact.}
Given the pebble game algorithms for sparsity in \cite{pebblegames}, 
we anticipate our result will lead to efficient approaches for determining the rigidity of body-and-cad
frameworks as well as finding rigid components and detecting minimally dependent sets of constraints. 
{\em Algorithms based on our characterization would provide the capability of giving 
valuable feedback to users designing highly complicated CAD systems.}

\medskip\noindent{\bf Related work.} 
Geometric constraint systems are at the core of constraint-based CAD software and have
been studied from several perspectives. 

In the CAD community, emphasis has been on finding solutions for realizing 
geometric constraint systems. A common approach is to use a {\em decomposition-recombination} scheme based on 
graph algorithms and numerical solvers; see, e.g., \cite{decompSurvey,sitharamSurvey}
for a survey of standard techniques. 
These decomposition schemes often rely on breaking a system into rigid (also called ``well-constrained'')
sub-systems. Difficulties arise even in studying small rigid sub-systems. 
The work of Gao et al. \cite{GaoHoffmannYangLocus04} 
enumerates and analyzes all ``basic configurations'' with up to six geometric
primitives, then presents a method for finding their solutions;
the authors observe that the most challenging configurations included
constraints involving lines (in contrast to those involving
only points or planes).

In classical rigidity theory, bar-and-joint structures composed
of universal joints connected by fixed length bars, i.e., point-point distance constraints, are studied.
For an overview of combinatorial rigidity theory, see the texts \cite{GraverCounting01,comb_rigidity}.

Combinatorial properties of rigidity are usually tied to {\em sparsity} counting conditions. A graph $G$ is $(k,\ell)$-{\em sparse} if the induced subgraph on any subset of $n'$ vertices contains at most $kn'-\ell$ edges.
 Under certain conditions, there are efficient algorithms to detect $(k, \ell)$-sparsity. If $G$ is $(k, \ell)$-sparse and $0 \leq \ell < 2k$, then its edges give rise to a matroid.  Within this matroidal range of values of $k$ and $\ell$, $(k,\ell)$-sparsity is determined
by a family of  pebble game algorithms which run in $O(n^2)$ worst case time \cite{pebblegames}.

In 2D, bar-and-joint rigidity is characterized by Laman's theorem with $(2,3)$-sparsity;
other characterizations followed, some based on decompositions of the edge set into trees 
\cite{lovasz:yemini,Re84,crapo:rigidity:88}. 
However, 3D bar-and-joint rigidity is not well-understood, and a combinatorial characterization 
is arguably the biggest open
problem in rigidity theory. While $(3,6)$-sparsity is necessary, it is not sufficient. Figure \ref{fig.doubleBanana} 
depicts the classical ``double banana'' counterexample; although it satisfies the necessary counts,
the structure is flexible, as the ``bananas'' can rotate about the dotted axis. 
\begin{figure}[tbh]
	\centering\includegraphics[width=.25\linewidth]{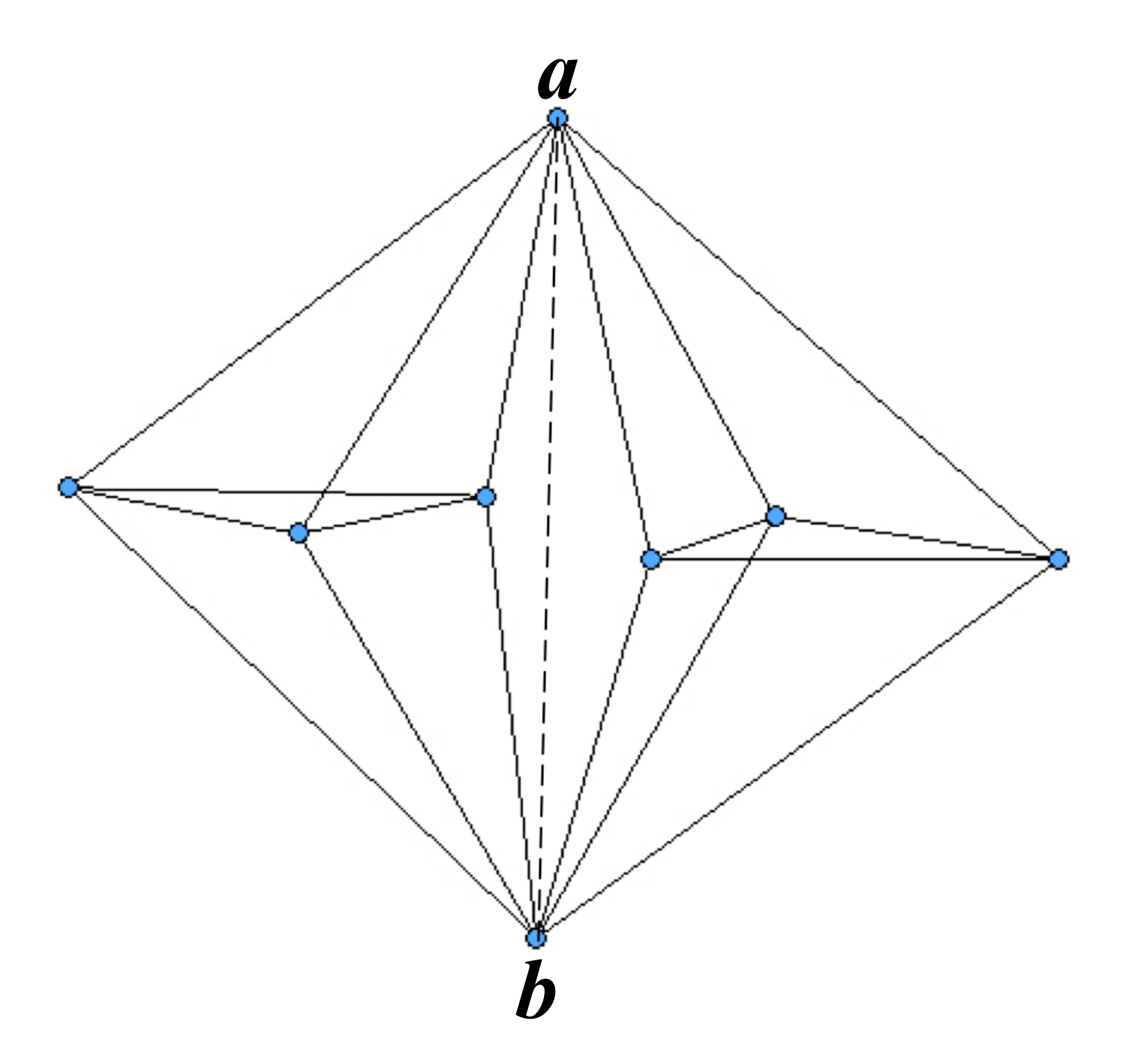}
	\caption{Double banana counterexample shows that $(3,6)$-sparsity is not sufficient for 3-dimensional bar-and-joint rigidity. The structure is flexible, as the ``bananas'' joined at points $a$ and $b$ can rotate about the dotted axis.}
	\label{fig.doubleBanana}
\end{figure}

While bar-and-joint structures in 3D are not well-understood, a closely related structure called the
body-and-bar structure is. 
The body-and-cad constraints that we consider include point-point distance constraints, which are simply bars -- the body-and-bar rigidity model is a special case of body-and-cad. Body-and-bar rigidity was characterized by Tay \cite{tayRigidity}; an alternate proof was given by White and Whiteley \cite{whiteWhiteley}. We generalize White and Whiteley's proof to obtain our result. Tay's theorem \cite{tayRigidity} states that $d$-dimensional body-and-bar rigidity is equivalent to an associated graph consisting of ${d+1 \choose 2}$ edge-disjoint spanning trees; as a consequence of Nash-Williams and Tutte's theorems 
\cite{nashWilliams1961,tutte1961}
on the equivalence of edge-disjoint spanning trees
and sparsity, $d$-dimensional body-and-bar rigidity is characterized by
$({d+1 \choose 2},{d+1 \choose 2})$-sparsity.

Other constraints motivated by CAD applications have been studied in rigidity theory.
Servatius and Whiteley give
a combinatorial characterization of rigidity for frameworks with direction (orientation of the vector defined by
two points with respect to a global coordinate frame) and distance constraints \cite{ServatiusWhiteley99}.
However, angular constraints among points, even in the plane, have proven very challenging; 
in \cite{SaliolaWhiteley04} Saliola and Whiteley
showed that deciding independence of circle intersection angles in the plane has the same complexity as determining
3D bar-and-joint rigidity. Angular constraints between lines and planes or rigid bodies, though, is well-understood
\cite{anglesCCCG09}.

The results from the rigidity theory community have not been directly applied in CAD research
partly because rigidity for CAD systems has not yet been fully characterized, and partly because of issues related to {\em genericity} (discussed in detail in Section \ref{sec:genericity}).  
Note that if a framework is rigid, then the set of all rigid frameworks with the same combinatorics is an open subset of its parameter space.
Michelucci and Foufou \cite{MichelucciFoufouWitness06} 
use this observation and present probabilistic algorithms which rely on analyzing a generic ``witness'' 
which shares the incidence constraints of the original system, but may have different parameters for the distance and angular constraints.
Even in the plane, where Laman's theorem characterizes bar-and-joint rigidity, the restriction of 
genericity poses a subtle challenge. 
The work of Jermann et al \cite{jnt-adg02} observes that coincidence 
constraints, when expressed as distance constraints with value 0, exhibit non-generic behavior;
a similar situation arises for parallel constraints. This is addressed through 
a notion of ``extended structural rigidity''; Jermann et al make assumptions similar to the
genericity conditions for the characterization we present in this paper.

The foundations of {\em infinitesimal body-and-cad rigidity theory} were introduced in \cite{haller:etAl:bodyCad:CGTA:2012}.  From the development of a rigidity matrix, a natural {\em nested sparsity} condition was identified as a necessary, but not sufficient condition for rigidity. 
The example given there highlighted the need for a better understanding of angular constraints.
The need for special treatment of angular constraints was implicitly observed
in \cite{gao:lin:zhang:CTreeDecomposition:2006}, but,
to the best of our knowledge, there has been no other work providing an explicit understanding of their behavior.
With both the flavor of 3-dimensional bar-and-joint and body-and-bar rigidity, the question remained as to where the difficulty of body-and-cad rigidity lay. The results we present here provide an answer for the majority of body-and-cad constraints,
but indicate that point-point coincidences may pose a significant challenge (see Section \ref{sec:point}).

\smallskip\noindent{\bf Structure.}
In Section \ref{sec:prelims}, we provide the preliminary concepts necessary for understanding the foundations
of infinitesimal body-and-cad rigidity and give an example to highlight the distinct behavior of angular constraints. 
In Section \ref{sec:genKGframes}, we develop the rigidity matrix for a more general combinatorial object called a $(k,g)$-frame, which 
is used to prove our main theorem in Section \ref{sec:combChar}. 
In Section \ref{sec:limitations}, we discuss the subtleties and limitations resulting from
genericity assumptions as well as the challenges posed by point-point coincidence constraints.
We conclude with future work in Section \ref{sec:conclusions}.

\section{Preliminaries}
\label{sec:prelims}
For a geometric constraint model, there are three levels of rigidity theory:  {\bf algebraic}, {\bf infinitesimal} and
{\bf combinatorial}. In {\bf algebraic rigidity theory}, a system of equations expressing the geometric constraints is
studied; a solution to this system corresponds to a {\em realization} of the 
structure\footnote{In Appendix \ref{apx:algebraicRT}, we
provide the formal definitions associated with body-and-cad algebraic rigidity theory, 
including the family of ``length'' functions required to describe the geometry of the system.}. 
In {\bf infinitesimal rigidity theory}, the
first-order behavior
of the (usually quadratic) system of equations expresses the constraints in terms of instantaneous motions and
a {\em rigidity matrix}. 
In {\bf combinatorial rigidity theory}, rigidity is defined terms of the rank of this rigidity matrix, so
that {\em generically} rigidity depends only on properties of a certain graph.
{\em In this paper, we are concerned only with infinitesimal and combinatorial
body-and-cad rigidity.}

The {\bf infinitesimal body-and-cad rigidity theory} was first presented in \cite{haller:etAl:bodyCad:CGTA:2012} and
relies on expressing constraints in the Grassmann-Cayley algebra through 
the rows of a {\em rigidity matrix}. For completeness, we provide an overview of the foundations developed in \cite{haller:etAl:bodyCad:CGTA:2012}.

\medskip\noindent{\bf Body-and-cad constraints.}
The body-and-cad frameworks that arise from CAD software are naturally 3-dimensional.
There are 21 coincidence, angular and distance constraints that can be placed between a pair of geometric elements (points, lines or planes) on rigid bodies, which we enumerate below for clarity:
\begin{itemize}
	\item {\bf Point-point constraints.} Coincidence, distance.
	\item {\bf Point-line constraints.} Coincidence, distance.
	\item {\bf Point-plane constraints.} Coincidence, distance.
	\item {\bf Line-line constraints.} Parallel, perpendicular, fixed angle,
	coincidence, distance.
	\item {\bf Line-plane constraints.} Parallel, perpendicular, fixed angle,
	coincidence, distance.
	\item {\bf Plane-plane constraints.} Parallel, perpendicular, fixed angle,
	coincidence, distance.
\end{itemize}

\medskip\noindent{\bf Infinitesimal constraint equations.}   To derive equations from the geometric constraints,
we must first consider instantaneous rigid body motion in $\RR^3$. As a consequence of Chasles' 
Theorem (see, e.g., \cite{selig2005}), any instantaneous rigid body motion may be described by
a {\em twist} (translation and rotation about a specified {\em twist axis}), 
which itself is represented by a 6-vector $\bs = (\bo, \bv)$.
The 3-vector $\bo$ describes the {\em angular velocity}: the direction of the twist axis and rotational
speed about it. The 3-vector $\bv$ can be used to decode the rest of the twist axis and translational speed along it. 

A {\em primitive} constraint between two bodies $i$ and $j$ is encoded by a single homogeneous 
linear equation on the twists $\bs_i = (\bo_i, \bv_i)$ and $\bs_j= (\bo_j, \bv_j)$. 
Each body-and-cad constraint is associated to a number of primitive constraints (which intuitively affect at most
one degree of freedom)\footnote{Four basic constraints, described in more detail in Appendix \ref{apx:fix2rows},
expressed 20 of the constraints; in Appendix \ref{apx:ptPtCoinc}, we discuss why we cannot address point-point constraints with our proof technique.}.
A distinction is made between primitive {\em angular} and {\em blind} constraints: a primitive angular 
constraint may affect only a rotational degree of freedom, while a primitive blind constraint may affect
either a rotational or translational degree of freedom.
Equations corresponding to angular constraints have zeroes in the coordinates corresponding to $\bv$.

The main contribution of \cite{haller:etAl:bodyCad:CGTA:2012} was the algebro-geometric derivation of these equations which 
are collected together in the
rigidity matrix to express the infinitesimal body-and-cad constraints. 
Given a structure with $n$ bodies, the matrix has $6n$ columns; these are
arranged so that the 6 columns for the $i$th vertex correspond
to $(\bv_i, -\bo_i)$\footnote{The re-ordering and negation are technicalities arising from the development 
of the rigidity matrix.}.
{\bf For the combinatorial characterization that we present in this paper, we are 
only concerned with the pattern of non-zero entries in the rigidity matrix}.

Elements of the kernel of  the rigidity matrix may be interpreted as the {\em infinitesimal motions} of the framework. 
If the only infinitesimal motions are {\em trivial}, i.e., assign the same twist to each body, the
framework is {\em infinitesimally rigid}; otherwise, it is {\em infinitesimally flexible}. 
Since we will only work in the infinitesimal rigidity theory, for brevity, we will drop ``infinitesimally'' 
for the remainder of this paper. A {\em minimally rigid} framework is one that is rigid, but becomes
flexible after the removal of any (primitive) constraint.

\smallskip\noindent{\bf Example.}	
	We provide a small example to illustrate the concepts of {\bf primitive} constraints and their
	further separation into {\bf angular} and {\bf blind} constraints. To emphasize that we are
	only concerned with the pattern of zero and non-zero entries, we use $\ast$-entries 
	to indicate values that are generically non-zero.
	
	The body-and-cad framework depicted in Figure \ref{fig.minRigBodyCad}
	is composed of two bodies (a cube and hexagonal prism) sharing three constraints:
	(i) a line-line coincidence, (ii) a point-plane coincidence and (iii) a point-point distance.
	This example is minimally rigid; for instance, removal of the (primitive blind) point-point distance constraint results
	in a flexible structure with one degree of freedom (see Figure \ref{fig.flexBodyCadExample}).
	
	The constraints for this framework are infinitesimally expressed via a 
	rigidity matrix with the pattern depicted in Figure \ref{fig.matrix}.
	Observe that the line-line coincidence constraint is associated to four rows of the matrix; this
	corresponds to four {\bf primitive} constraints. In fact, rows 1 and 2 express a line-line parallel
	constraint and have 0 values in three columns associated with each body; each of these rows is expressing
	a {\bf primitive angular} constraint. The remaining four rows express {\bf primitive blind} constraints.
	\begin{figure}[bth]
		\begin{center}
		\begin{tabular}{ccccc}

			 & $\bv_1$& $-\bo_1$ &$\bv_2$& $-\bo_2$ \\
		\cline{1-5}
		\multirow{4}{25mm}{line-line coincidence}&\multicolumn{1}{|c|}{\cellcolor{red}000}&\multicolumn{1}{c|}{\cellcolor{red}$\ast\ast\ast$}&\multicolumn{1}{c|}{\cellcolor{red}$000$ }&\multicolumn{1}{c|}{\cellcolor{red}$\ast\ast\ast$}\\
		\cline{2-5}
		&\multicolumn{1}{|c|}{\cellcolor{red}000}&\multicolumn{1}{c|}{\cellcolor{red}$\ast\ast\ast$}&\multicolumn{1}{c|}{\cellcolor{red}$000$ }&\multicolumn{1}{c|}{\cellcolor{red}$\ast\ast\ast$}\\
		\cline{2-5}
		&\multicolumn{1}{|c|}{$\ast\ast\ast$ }&\multicolumn{1}{c|}{$\ast\ast\ast$}&\multicolumn{1}{c|}{$\ast\ast\ast$}&\multicolumn{1}{c|}{$\ast\ast\ast$}\\
		\cline{2-5}
		&\multicolumn{1}{|c|}{$\ast\ast\ast$ }&\multicolumn{1}{c|}{$\ast\ast\ast$}&\multicolumn{1}{c|}{$\ast\ast\ast$}&\multicolumn{1}{c|}{$\ast\ast\ast$}\\
		\cline{1-5}
	 	\multicolumn{1}{c|}{point-plane distance}&\multicolumn{1}{c|}{$\ast\ast\ast$ }&\multicolumn{1}{c|}{$\ast\ast\ast$}&\multicolumn{1}{c|}{$\ast\ast\ast$}&\multicolumn{1}{c|}{$\ast\ast\ast$}\\
			\cline{1-5}
		\cline{1-5}
		\multicolumn{1}{c|}{point-point distance}&\multicolumn{1}{c|}{$\ast\ast\ast$ }&\multicolumn{1}{c|}{$\ast\ast\ast$}&\multicolumn{1}{c|}{$\ast\ast\ast$}&\multicolumn{1}{c|}{$\ast\ast\ast$}\\
			\end{tabular}
		\end{center}
		\caption{The pattern of the rigidity matrix for the structure depicted in Figure \ref{fig.minRigBodyCad}.}
		\label{fig.matrix}
	\end{figure}

\medskip\noindent{\bf Novelty.} We emphasize the two main differences between the 
development of the infinitesimal rigidity theory for the body-and-cad \cite{haller:etAl:bodyCad:CGTA:2012} 
and body-and-bar \cite{tayRigidity,whiteWhiteley} models. In the rigidity matrix associated to a {\bf body-and-bar} 
framework, each row corresponds to a bar, or equivalently, a distance constraint between two bodies.  By contrast, 
in order to encode a single constraint in a {\bf body-and-cad} framework, we may need {\em multiple rows} in the 
rigidity matrix; the concept of primitive constraints addresses this by associating exactly one constraint with each row. 

The second difference appears in the separate treatment of primitive angular and blind constraints, 
represented by red and black edges, respectively, in a primitive cad graph, defined below. This was the main challenge
left open by the natural {\em nested sparsity} condition identified as a necessary, but not sufficient, condition
for body-and-cad rigidity. The example given on page 27 of \cite{haller:etAl:bodyCad:CGTA:2012} highlighted the need 
for a better understanding of the interplay of angular and blind constraints.

\subsection{The primitive cad graph}\label{sec:primitive}
The infinitesimal theory of \cite{haller:etAl:bodyCad:CGTA:2012}
lets us associate a {\em primitive cad graph} $H= (V, B \disjointUnion R)$ to each body-and-cad framework. 
Here, $H$ is a multigraph with a vertex for each body and an edge for each primitive constraint.   
Moreover, the edges are partitioned into two sets, a set $B$ of black edges and a set $R$ of red edges, 
representing primitive blind and angular constraints, respectively.  See Figure \ref{fig.minRigBodyCadPrim} for 
the primitive cad graph associated to the minimally rigid framework in Figure \ref{fig.exampleBodyCad}. 
\begin{figure}[bth]	
	\centering
	\includegraphics[width=.3\linewidth]{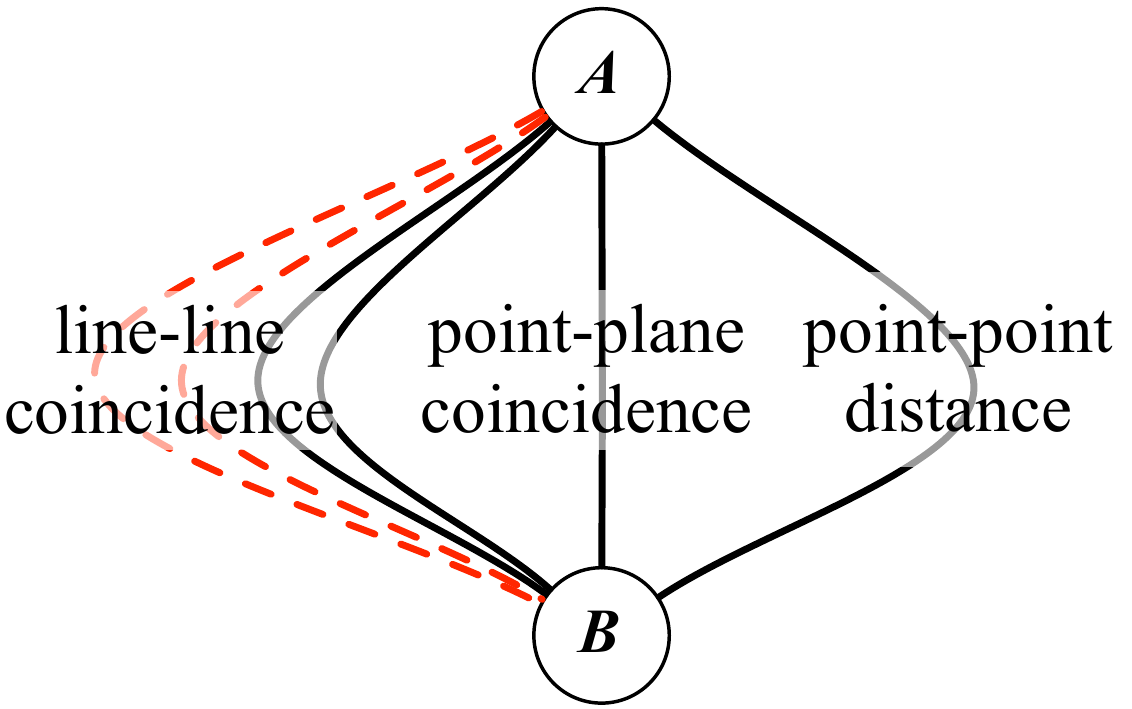}
	\caption{The primitive cad graph associated to the framework from Figure \ref{fig.minRigBodyCad} has 2 red (dashed) edges, representing primitive angular constraints, and 4 black (solid) edges, representing primitive blind constraints.}
	\label{fig.minRigBodyCadPrim}
\end{figure}

Recall that each primitive constraint is expressed as a single linear equation. In fact, this equation can be encoded by a 
6-vector whose entries are derived from the specific geometry of the framework.
We define a {\em primitive cad framework} $H(\bp)$ to be the graph $H$ in which each edge is labeled with this 6-vector. 
 Let $m$ be the number of edges in $H$ and $m_R=|R|$ be the number of red edges. Our main theorem for 3-dimensional body-and-cad rigidity is stated in terms of the primitive cad graph $H$.

\begin{theorem}\label{theorem:combchar}
	A 3-dimensional body-and-cad framework (in which no point-point coincidence constraints are present) is generically minimally rigid if and only if, in its associated primitive cad graph $H = (V, R \disjointUnion B)$, 
	there is some set of black edges $B' \subseteq B$ such that
\begin{enumerate}
\item  $B \backslash B'$ is the edge-disjoint union of $3$ spanning trees, and
\item  $R \cup B'$ is the edge-disjoint union of $3$ spanning trees.
\end{enumerate}
\end{theorem}

\smallskip\noindent Body-and-cad frameworks are a special case of the more general $(k,g)$-frames introduced in 
Section \ref{sec:genKGframes}. Theorem \ref{theorem:combchar} follows from a result (Theorem \ref{theorem:minRigIFFTrees}) 
in this more general setting, which we prove in Section \ref{sec:combChar}.

\section{From body-and-cad frameworks to $(k,g)$-frames}
\label{sec:genKGframes}
We begin in Section \ref{sec:rigidityMatrix} by introducing the $(k,g)$-frame, which generalizes both the notion 
of a body-and-cad framework and the $k$-frame of \cite{whiteWhiteley}, and its associated rigidity matrix. 
We then present a pure condition for the minimal rigidity of a $(k,g)$-frame in Section \ref{sec:pureCondition}. 
From the point of view of applications, our result is most interesting for 3-dimensional body-and-cad frameworks; 
however, the techniques do not depend on dimension, so we work in the most general setting. 

\subsection{$(k,g)$-frames}
\label{sec:rigidityMatrix}
The proof of our main result depends only on the pattern of zeroes in the rigidity matrix of a body-and-cad framework, 
and not on their derivation or interpretation.   Thus, we may view the pattern in the rigidity matrix associated to a 
primitive cad framework as an instance of a rigidity matrix associated to a more general $(k,g)$-frame, 
which we introduce below.   We begin with some preliminary combinatorial definitions.
 
\begin{defn} A multigraph $H = (V,E=B \disjointUnion R)$ is a {\em bi-colored graph}, where $V = \{1, \ldots, n\}$ is the vertex set,
and $E$ is a set of $m$ edges, $m_B$ of which are black and $m_R$ of which are red.  For brevity, we use ``graphs'' to refer to multigraphs.
\end{defn} 

We generalize the notion of a $k$-frame in Definition 1.7 of \cite{whiteWhiteley}. 
\begin{defn}
	\label{def:kgframe}
A {\em $(k,g)$-frame} $H(\bp)$ is a bi-colored graph $H = (V, E = B \disjointUnion R)$ 
together with a function $\bp: E \to \RR^k$, where $\bp(r)_j = 0$ for $j =1, \ldots, k-g$ and $r \in R$.  
\end{defn}
Then, primitive cad frames in 3D are simply $(6,3)$-frames. 
The $k$-frames of White and Whiteley are $(k,g)$-frames in which $R$ is empty (hence the value of $g$ is irrelevant). 

Associated to each $(k,g)$-frame is a rigidity matrix which we define below.

\begin{defn}\label{def:rigMatrix}
Given a $(k,g)$-frame $H(\bp),$ we define a matrix $M(H(\bp))$ with $k$ columns per vertex.  
For each edge $e = uv$ with $u<v$, define a row with $\bp(e)$ in the $k$ columns for $u$ and $-\bp(e)$ in the $k$ columns for $v$. 
\end{defn}

Motions assign a vector of length $k$ to each body; for structures in dimension $d$, 
$k = {d+1 \choose d}$, and each vector represents a compatible infinitesimal motion.
\begin{defn}
	A {\em motion} of a $(k,g)$-frame is a vector of length $kn$ that is orthogonal to the
	row space of $M(H(\bp))$. A {\em trivial motion} assigns the same $k$-vector to each body, i.e.,
	has $n$ copies of the same $k$-vector.
\end{defn}
We are concerned with {\em minimal rigidity}, i.e., structures with no redundant constraints. 
\begin{defn}
A $(k,g)$-frame $H(\bp)$ is {\em minimally rigid}\footnote{Also called ``isostatic'' in the literature.}
	if it becomes flexible after the removal of any edge from $H$.
\end{defn}

\begin{rmk} Definitions \ref{def:kgframe} and  \ref{def:rigMatrix} set the pattern of zeroes in the matrix $M(H(\bp)).$
However, the rows of the body-and-cad rigidity matrix presented in \cite{haller:etAl:bodyCad:CGTA:2012} are more specialized.  
In Appendix \ref{apx:fix2rows} we give an alternative derivation of equations for two of the basic body-and-cad constraints 
so that the equations conform to the pattern we have set.
\end{rmk}

\subsection{The pure condition for minimal rigidity}
\label{sec:pureCondition}
As with other classical rigidity results, our characterization relies solely on the {\em combinatorics} of
a body-and-cad framework, which applies \emph{generically}.  In this section we introduce the pure condition 
for the minimal rigidity of a $(k,g)$-frame, generalizing the setup in \cite{whiteWhiteley}.  In fact, the existence of a 
pure condition is what justifies our genericity assumptions, which we discuss in more detail in Section \ref{sec:genericity}.

The definition below generalizes necessary edge-counting conditions for minimal rigidity, taking angular constraints into account.
\begin{defn}
	Let $k$ and $g$ be positive integers.  A bi-colored graph $H$ is {\em $(k,g)$-counted} if 
	\begin{enumerate} 
		\item $m = kn-k$ (count on total edges), and 
		\item $m_R \leq gn-g$ (count on red edges).
	\end{enumerate}
\end{defn}

The matrix $M(H(\bp))$ of a $(k,g)$-counted $(k,g)$-frame is not square
(it has $kn-k$ rows and $kn$ columns). It is straightforward to observe that any
minimally rigid $(k,g)$-frame must have $kn-k$ rows
since the trivial motions will always be in the kernel of $M(H(\bp))$.  If we fix or ``tie down" one of the 
bodies in the framework, then these trivial motions are no longer allowed.  

\begin{defn}
	The {\em basic tie-down} of a $(k,g)$-counted  $(k,g)$-frame is a $k \times kn$ matrix $T(k)$ with the identity matrix
	in the first $k$ columns and 0s everywhere else.

	Given a $(k,g)$-counted $(k,g)$-frame $H(\bp)$, the (square) matrix $M_T(H(\bp))$ is $M(H(\bp))$ with $T(k)$ 
appended to the bottom. 
\end{defn}

Thus, the rigidity of  a $(k,g)$-frame $H(\bp)$ satisfying the necessary counts may be determined by 
computing the determinant of the matrix $M_T(H(\bp))$.

\begin{prop}[Proposition 2.5 from \cite{whiteWhiteley}]
\label{prop:minRigDet}
	A $(k,g)$-counted $(k,g)$-frame $H(\bp)$ is minimally rigid if and only if
	$M_T(H(\bp)) \neq 0$.
\end{prop}

In fact, since we would like to understand rigidity more generally than for a single set of values for $\bp,$ 
we would like to think of $M_T(H(\bx))$ as a polynomial in indeterminates.  As in \cite{whiteWhiteley}, we make the following definition.

\begin{defn}
Given a $(k,g)$-frame $H(\bp)$, we define the \emph{generic} $(k,g)$-frame $H(\bx)$ on the underlying labeled graph $H$ by 
setting all generically non-zero entries of the $k$-vectors labeling the edges to be algebraically independent 
indeterminates $\bx(b)_j$ for $j=1,\ldots,k$, and $\bx(r)_j$ for $j=k-g+1,\ldots,k$, where $b\in B$ and $r \in R$. 
\end{defn}

We define the \emph{pure condition} to be $C_H(\bx) = \det M_T(H(\bx))$ so that $C_H(\bx)$ is a polynomial 
in $N = km_B+gm_R$ variables. Generic rigidity is expressed in terms of the pure condition.
\begin{defn} 
	\label{defn:genericRigidity}
	A $(k,g)$-counted $(k,g)$-frame $H(\bx)$ is {\em generically minimally rigid} if there exists a function 
	$\bp: E \to \RR^k$ for which $C_H(\bp) \neq 0$ satisfying $\bp(r)_i=0$ for $i = 1, \ldots, k-g$ for all $r \in R$.
\end{defn}

We require one final combinatorial concept for analyzing the pure condition. 
\begin{defn}
	A {\em $(k,g)$-fan} $\phi$ of a $(k,g)$-counted bicolored graph
	is a partitioning of $E$ into $n-1$ ordered sets $(\phi_2, \ldots, \phi_n)$
	such that 
	\begin{enumerate}
		\item each $\phi_i = (\phi_{i,1}, \ldots, \phi_{i,k})$ contains exactly $k$ edges incident to vertex $i$, and
		\item each $\phi_i$ contains $\leq g$ (red) edges from $R$.
	\end{enumerate}
\end{defn}
\noindent We denote
the determinant of the matrix whose rows are $\bx(\phi_{i,1}), \ldots, \bx(\phi_{i,k})$  by $[\phi_i]$. 

\begin{defn}
	Two $(k,g)$-fans $\phi$ and $\phi'$ are called {\em distinct} if there exists a vertex $i$
such that $\phi_i \neq \phi_i'$, as unordered sets. 
\end{defn}
\noindent We can represent distinct $(k,g)$-fans by orienting $H$: if edge $e \in \phi_i,$ orient 
it so that its tail is at vertex $i$.
Then, given a $(k,g)$-fan $(\phi_2,\ldots,\phi_n)$, its {\em $(k,g)$-fan diagram}
	is the oriented multigraph 
	$F = (V,A)$, where $\overrightarrow{ij} \in A$ if and only if $ij \in \phi_i$.
Note that the same $(k,g)$-fan diagram will represent multiple $(k,g)$-fans,
as depicted in Figure \ref{fig:kfan}, but exactly one distinct $(k,g)$-fan. 

\begin{figure}[tbh]
\centering \subfloat[Example $(3,1)$-counted bicolored graph has 9 total edges, 2 of which are red.] {\label{fig:thicket31labeled}
\begin{minipage}[b]{.3\linewidth}
\centering
\includegraphics[width=.7\linewidth]{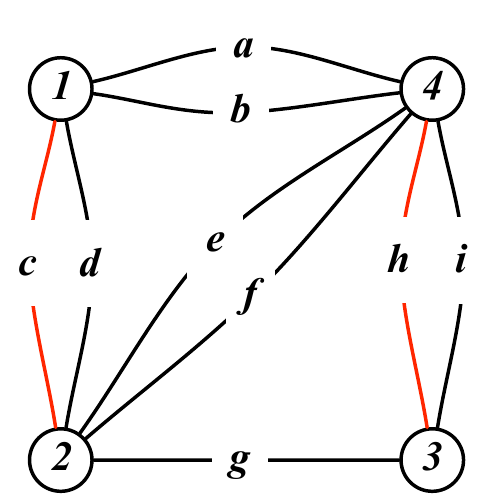}
\end{minipage}}%
\hspace{3mm}
\centering \subfloat[This $(3,1)$-fan diagram corresponds to multiple
$(3,1)$-fans, including: ((c,d,e),(g,h,i),(a,b,f)) and
((e,c,d),(i,g,h),(f,a,b)).
]{\label{fig:kfan}
\begin{minipage}[b]{0.3\linewidth}
\centering\includegraphics[width=.7\linewidth]{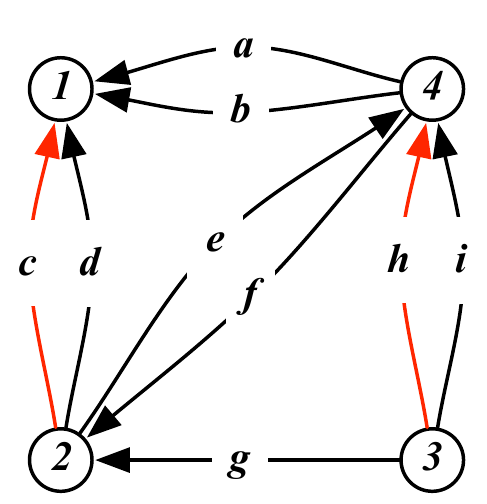}
\end{minipage}}
\hspace{3mm}
\centering \subfloat[Black edges without bolded edge form 2 edge-disjoint 
spanning trees (dotted and dashed); red edges with bolded edge form a single 
spanning tree (solid).]{\label{fig:thicket31trees}
\begin{minipage}[b]{0.3\linewidth}
\centering\includegraphics[width=.65\linewidth]{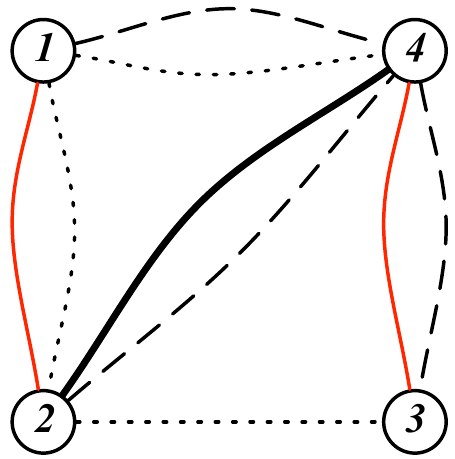}
\end{minipage}}
\caption{An example $(3,1)$-counted graph satisfying the conditions of Theorem \ref{theorem:minRigIFFTrees}. The graph has two red edges: $c$ between vertices 1 and 2, and $h$ between vertices 3 and 4.}
\label{fig:thicket31}
\end{figure}

As a consequence of Proposition 2.12 from \cite{whiteWhiteley}, we obtain the following.  
\begin{prop}
	\label{prop:pureCondKGFan}
	$C_H(\bx) = \Sigma_{\phi}\pm  [\phi_2]\cdots[\phi_n]$, for all distinct $(k,g)$-counted $(k,g)$-fans $\phi$ of $H$.
\end{prop}
\begin{proof}
	If we compute $\det M_T(H(\bx))$ via a Laplace expansion on the first $k$ columns of $M_T(H(\bx))$, we see that the determinant
	is actually the determinant of the $(kn-k) \times (kn-k)$ matrix formed by deleting the first $k$ columns and last $k$ rows.  We can in
	turn compute the determinant of this submatrix via a Laplace expansion taking the $k$ columns of each vertex at a time.  This requires
	that we compute the determinant of matrices formed by partitioning the rows of our matrix so that $k$ rows are associated to each vertex $2, 
	\ldots, n.$  Each row corresponds to an edge, and if we have a submatrix containing more than $g$ rows associated to red edges, then the 
	corresponding determinant is zero.  Therefore, $C_H(\bx)$ is the sum of terms of the form  $\pm [\phi_2]\cdots[\phi_n]$, where $\phi=(
	\phi_2, \ldots, \phi_n)$ is a distinct $(k,g)$-fan.  The sign of each term is determined by the rules for Laplace expansion, the permutation $\phi$,
	and whether the vectors used in the Laplace expansion are $\bx(e)$ or $-\bx(e).$	
\end{proof}

\section{A combinatorial characterization of minimal rigidity}
\label{sec:combChar}
Given the setup of \cite{haller:etAl:bodyCad:CGTA:2012} and the previous sections, our main result 
may be stated and proved in a purely combinatorial setting. Theorem~\ref{theorem:minRigIFFTrees} and its proof 
are a generalization of the body-and-bar characterization theorem of White and Whiteley.
\begin{theorem}
	\label{theorem:minRigIFFTrees}
	A $(k,g)$-counted $(k,g)$-frame $H(\bx)$ is generically minimally rigid if and only if there exists some set of
	black edges $B' \subseteq B$ such that
	\begin{enumerate}
		\item $B \setminus B'$ is the edge-disjoint union of $k-g$ spanning trees, and
		\item $R \cup B'$ is the edge-disjoint union of $g$ spanning trees.
	\end{enumerate}
\end{theorem}

\noindent Figure \ref{fig:thicket31trees} depicts an example of a $(3,1)$-counted bi-colored graph
satisfying the conditions of Theorem \ref{theorem:minRigIFFTrees}.

\begin{proof}
$(\Longrightarrow)$ 
Assume $C_H(\bx) \neq 0$.
Consider a Laplace expansion along the last $k$ rows of $M_T(H(\bx))$. Since the only non-zero entries in these rows
occur in the $k \times k$ identity matrix appearing in the tie-down in the first $k$ columns, $\det M_T(H(\bx)) =\det A$, 
where $A$ is the submatrix of $M_T(H(\bx))$ formed by the first $kn-k$ rows and last $kn-k$ columns. 
Since $C_H(\bx) = \det M_T(H(\bx))$, then $\det A \neq 0$.

Now consider a Laplace expansion of $\det A$ using $(n-1)\times (n-1)$-minors so that $\det A$ is a sum of 
terms $\mu =  \det A_1\cdots \det A_k,$ where for $j = 1,\ldots, k$, $A_j$ is an $(n-1)\times (n-1)$ submatrix 
of $A$ using the $j$th columns associated to each of the $n-1$ remaining vertices and some choice of $n-1$ rows.  
If $\det A$ is non-zero, then some term $\mu= \det A_1\cdots \det A_k$ is non-zero, which implies that 
$\det A_1, \ldots, \det A_k$ are all non-zero.  Each submatrix $A_j$ has one column per vertex, and the rows of 
$A_j$ are just the rows of the incidence matrix of an orientation on the edges of $H$ multiplied by non-zero scalars.  
Since $\det A_j$ is non-zero, $A_j$ is the incidence matrix of a subgraph of $H$ with $n-1$ edges and no cycles.  
Hence, $A_j$ describes a spanning tree $T_j$ of $H.$

If $r$ is a red edge, then $\bx(r)_j=0$ for $j = 1,\ldots,k-g$.  Since $\det A_j$ is non-zero for all $j$,  
the edge $r$ must be in one of the trees $T_{k-g+1}$,\ldots,$T_k$. Let $B'\subseteq B$ be the set of black edges 
in $T_{k-g+1},\ldots,T_k$.  Then $B \setminus B'$ is the edge-disjoint union of $T_1,\ldots,T_{k-g}$
and $R \cup B'$ is the edge-disjoint union of the $g$ remaining spanning trees.

\medskip
$(\Longleftarrow)$ 
Let $H(\bx)$ be the generic $(k,g)$-frame, where $H = (V,E = B \disjointUnion R)$ is a 
$(k,g)$-counted graph,
and let $B' \subseteq B$ be a set of edges such that 
\begin{enumerate}
	\item $B \setminus B'$ is the edge-disjoint union of $k-g$ spanning trees $T_1,\ldots,T_{k-g}$,
	\item $R \cup B'$ is the edge-disjoint union of $g$ spanning trees $T_{k-g+1},\ldots,T_k$.
\end{enumerate} 
We show that $\det M_T(H(\bx))$ is not identically zero.  To do this it suffices to show that there exists a 
function $\bx':E \to \RR^k$ such that $\det M_T(H(\bx')) \neq 0$. For $j = 1, \ldots, k$ and $i = 1, \ldots, k$, 
let $a_{j,i}$ denote an indeterminate.  For each $j$ and $e \in T_j$ let 
$$\bx'(e) = \left\{ \begin{array}{ll} 
						(a_{j,1},a_{j,2},\ldots,a_{j,k}) & \hbox{if } e \in B \setminus B' \\
						(0,0,\ldots,0,a_{j,k-g+1},\ldots,a_{j,k}) & \hbox{if } e \in R \cup B'						
						\end{array}\right .$$
						
By Proposition \ref{prop:pureCondKGFan}, $C_H(\bx) = \Sigma_{\phi}\pm \Pi_{i=2}^n [\phi_i]$. We claim that the only non-zero
term in the pure condition stems from a single distinct $(k,g)$-fan $\phi$.  Root each tree at vertex 1.  Each tree has a 
unique path from vertex $i$ to the root, so we can orient all edges toward the root.  Define $\phi_i$ to be the ordered set 
of all edges incident to vertex $i$ pointing toward the root, where the $j$th edge in $\phi_i$ is the edge in $T_j.$   
By construction, each $\phi_i$ has at most $g$ red edges, so $[\phi_i]\neq 0.$ 

Consider another distinct $(k,g)$-fan $\phi' \neq \phi$. There must be some vertex $i$ such that $\phi'_i$ contains at least two
edges $e_1$ and $e_2$ from the same tree $T_j$. Hence, $\bx'(e_1) = \bx'(e_2)$
and $[\phi'_i] = 0$, and we conclude that $\phi$ is the only fan for which $\Pi_{i=2}^n [\phi_i]$ is non-zero. 
\end{proof}

We now prove Theorem \ref{theorem:combchar}, which gives a combinatorial characterization for 3-dimensional body-and-cad rigidity (omitting point-point coincidence constraints).

\begin{proof}[Proof of Theorem \ref{theorem:combchar}]
The result follows from Theorem \ref{theorem:minRigIFFTrees} as a primitive cad graph for a body-and-cad framework
with no point-point coincidence constraints is a $(k,g)$-frame with $k=6$ and $g=3.$
\end{proof}

\section{Considerations and challenges}
\label{sec:limitations}
Before concluding, we discuss practical implications of the genericity assumptions that result from the pure
condition, then provide context as to why point-point coincidence constraints pose a challenge.

\subsection{Genericity}
\label{sec:genericity}
Like all combinatorial characterizations of rigidity, our main theorem gives
a condition for {\em generic} minimal rigidity. This notion of genericity is subtle as there are currently
no geometric conditions implying or characterizing it (even in the most well-studied bar-and-joint rigidity model). 
Common assumptions in computational geometry, such as the general position of joint coordinates, are not strong enough to ensure the 
genericity of a bar-and-joint structure, and we have additional types of constraints to consider. Thus it is difficult to test when a 
given framework satisfies the genericity assumptions needed for the combinatorial analysis.

We defined generic rigidity using the pure condition (see Definition \ref{defn:genericRigidity}) and now
discuss why this definition is appropriate. Recall that $N = km_B+gm_R$. The set of all points $\bq \in \RR^N$ with $C_H(\bq) = 0$ 
is a closed subset that either has dimension less than $N$ or is all of $\RR^N$.  Consequently, the set of points $\bq \in \RR^N$ 
such that $C_H(\bq) \neq 0$ is open.  This implies that, if there exists one $\bq \in \RR^N$ such that $C_H(\bq) \neq 0,$ then the 
pure condition is non-zero for all points in some small neighborhood of $\bq$.   Moreover, if $C_H(\bq) \neq 0$ for some $\bq,$ 
then the set of points with $C_H(\bq)=0$ is a \emph{proper} closed subset of $\RR^N$ with measure zero, and $C_H(\bq')\neq 0$ 
\emph{generically}, i.e., for almost all values $\bq' \in \RR^N.$

Therefore, Definition \ref{defn:genericRigidity} can be viewed as an extension of Tay's notion of generic rigidity \cite{tayRigidity}
 to $(k,g)$-frames.  A $(k,g)$-frame $H(\bx)$ is generically minimally rigid if and only if $C_H(\bx)$ is not identically zero. 
If, however, $C_H(\bx)$ is not the zero polynomial and $C_H(\bp)=0$ for some $\bp \in \RR^N$, then 
$H(\bp)$ is (infinitesimally) flexible. We refer to $\bp$ as a {\em non-generic} point, and say that we have a non-generic 
embedding of the $(k,g)$-frame $H(\bp)$. If the polynomial $C_H(\bx)$ is identically zero, then $H(\bx)$ is {\em generically flexible}.

\smallskip\noindent{\bf A non-generic framework: Pappus's Theorem}\\
As observed by Jermann et al \cite{jnt-adg02}, there are inherent problems in analyzing the geometry of a specific realization of a 
framework using only the combinatorics of a graph, as a particular embedding or framework may fail to be generic.  

This problem is straightforward to observe in the analysis of body-and-bar frameworks,
in which all constraints are point-point distance constraints.  
For example, consider a graph with vertices $v_1$ and $v_2$ joined by 6 edges; 
this is the combinatorial model for two rigid bodies joined by 6 bars.  
If the bars are attached at points that are chosen generically, then the framework is rigid, matching
the combinatorial analysis provided by Tay's Theorem  
\begin{figure}[bth]
	\centering \subfloat[The 6 points $A,B,C,a,b,c$ in the non-generic position from the hypothesis of Pappus's Theorem.] {\label{fig:pappusDiag}
	\begin{minipage}[b]{.46\linewidth}
	\centering
	\includegraphics[width=\linewidth]{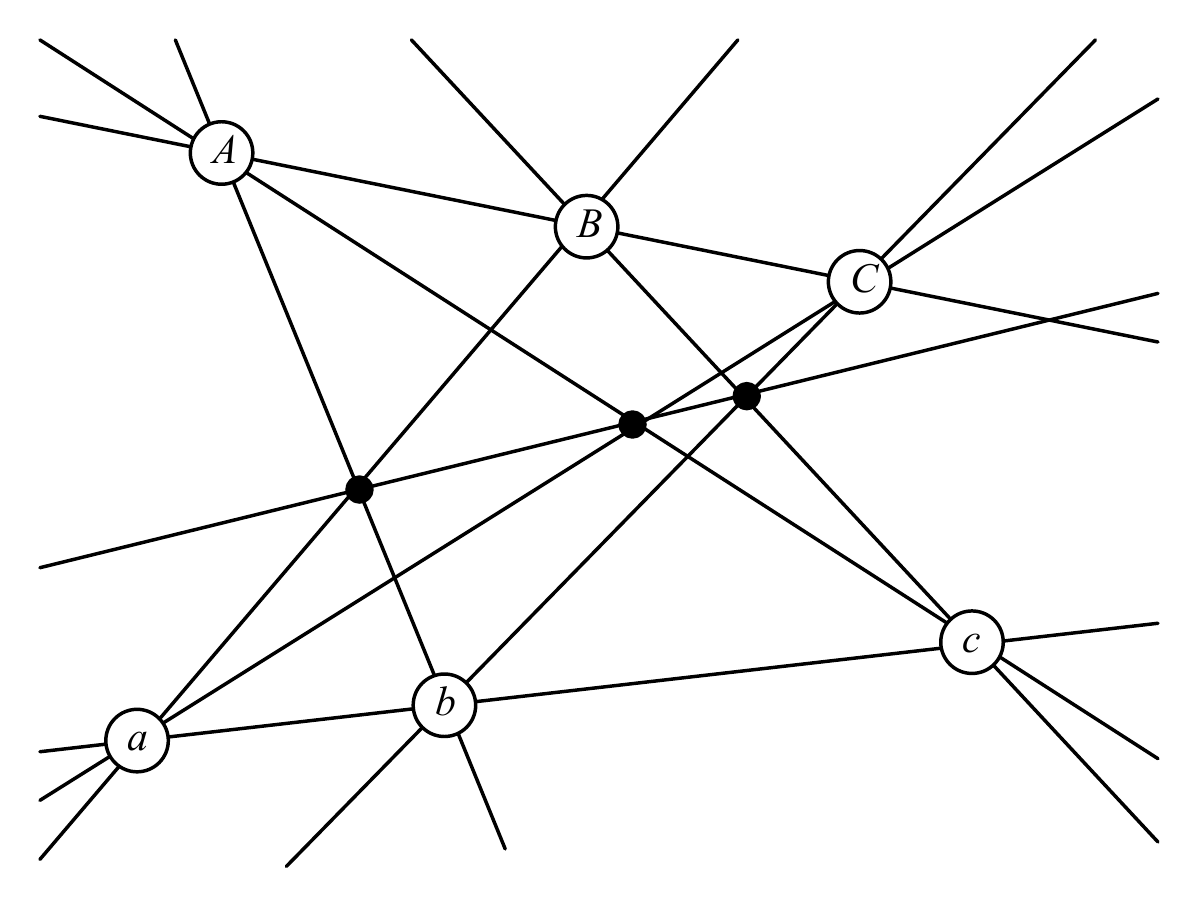}
	\end{minipage}}%
	\hspace{4mm}
	\centering \subfloat[Circular vertices are associated to the bodies from points and square vertices are associated to the bodies from lines. 
	Vertices are connected if there is a point-line coincidence between the corresponding rigid bodies. There are two edges for each 
	point-line constraint in the primitive cad graph, but we have drawn only one to aid readability. The vertex labeled $\cap$ {\em line} 
	represents the line $\overline{p_{ab}p_{bc}}$.] {\label{fig:pappusCad}
	\begin{minipage}[b]{.46\linewidth}
	\centering
	\includegraphics[width=\linewidth]{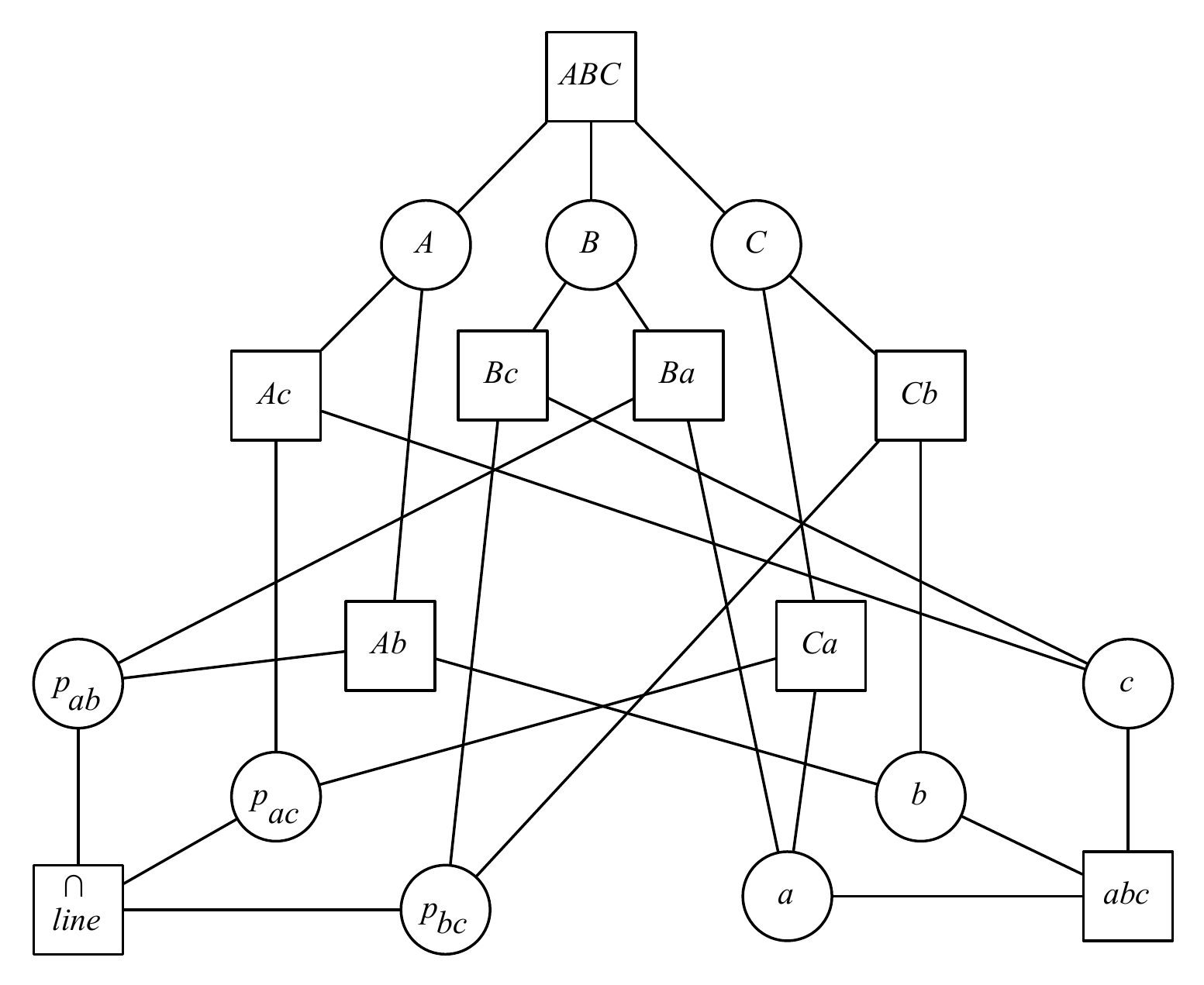}
	\end{minipage}}	
	\caption{The point-line coincidences described by Pappus's Theorem.}
	\label{fig.pappus}
\end{figure}
\cite{tayRigidity}. However, if the points are not sufficiently generic -- for example if all 6 bars attach to a 
single point on one of the bodies -- then the framework is flexible; for this {\em non-generic} embedding,
Tay's theorem does not apply. Intuitively, ``reusing'' the same point results in a dependency
that is not encountered in the generic case.

The constraints in a body-and-cad framework present similar challenges in applying combinatorial
analysis. To illustrate this, we present a body-and-cad framework along with 
a non-generic embedding, modeling Pappus' Theorem, in which a generically independent constraint becomes dependent. 

\begin{theorem}[Pappus] (See Figure \ref{fig:pappusDiag}.)
Let $A,B,C$ and $a,b,c$ be two sets of 3 collinear points in the plane.  Then the 3 points $p_{ab} = \overline{Ab} \cap \overline{Ba}$,  $p_{ac} = \overline{Ac} \cap \overline{Ca}$, and $p_{bc} = \overline{Bc} \cap \overline{Cb}$ are collinear.
\end{theorem}

To model Pappus's Theorem, we use a framework with 18 bodies: 9 to model the points,
and another 9 to model the two lines $\overline{ABC}$ and $\overline{abc}$, the six lines 
$\overline{Ab},\overline{Ac},\overline{Ba},\overline{Bc},\overline{Ca},\overline{Cb}$, and the
line $\overline{p_{ab}p_{bc}}$.  If we encode all of the point-line coincidences depicted in 
diagram (a) in Figure \ref{fig.pappus}, then the point-line coincidence between
$p_{ac}$ and $\overline{p_{ab}p_{bc}}$ is implied by the other constraints as a consequence of Pappus's Theorem.
Thus, there is a \emph{dependency} among the geometric constraints which should correspond to a dependency in the rigidity matrix.

However, a combinatorial analysis of the associated primitive cad graph (depicted schematically in Figure \ref{fig:pappusCad})
based on Theorem \ref{theorem:combchar} indicates that the constraints are {\em generically}
independent. Indeed, we can create a generic embedding of the same primitive cad graph
in which each constraint uses unique geometric elements: 
for each point-line coincidence, identify a distinct point and distinct line on the pair
of constrained bodies. If the 27 points and 27 lines are sufficiently generic,
then the associated rigidity matrix has full rank, and the constraints are generically independent.

\subsection{Point-point coincidence constraints} \label{sec:point}
The result given in this paper does not address point-point coincidence constraints.  In addition to 
the algebraic difficulties (described in Appendix \ref{apx:ptPtCoinc}) that prohibit the inclusion of point-point coincidence constraints
in the definition of a $(k,g)$-frame, point-point coincidences seem inherently more challenging to address.  

It may be that characterizing 3D point-point coincidences is just as challenging as 
understanding 3D bar-and-joint rigidity. Indeed, we can construct an
analogue of the classical ``double banana" example with a body-and-cad framework 
that has two rigid bodies joined by two point-point coincidence constraints. The two rigid bodies are free to rotate about the line 
joining their two points of intersection (labeled $\vec a$ and $\vec b$ in Figure \ref{fig.doubleBanana}), 
so the body-and-cad framework is flexible.  

In this example, each point-point coincidence constraint corresponds to three primitive blind constraints (see Appendix \ref{apx:ptPtCoinc}), 
and the primitive cad graph contains 2 vertices and 6 black edges between them. 
A natural extension of Theorem \ref{theorem:combchar} would characterize generic minimal rigidity with a partitioning of the edge set into 
6 edge-disjoint spanning trees (there are no red edges). While this property is necessary, it is not sufficient.
The primitive cad graph for the ``double banana'' satisfies the condition, but is flexible, thus serving as a counterexample.
Therefore, a combinatorial characterization of \emph{body-and-cad} frameworks including point-point coincidences must differ 
from the characterization for body-and-cad frameworks presented here.

\section{Conclusions and Future Work}
\label{sec:conclusions}
We have presented a combinatorial characterization of body-and-cad rigidity of
structures using 20 of the 21 pairwise constraints; we exclude point-point coincidence
constraints. This characterization is expressed in terms of the edge set of an
associated graph being partitioned into two sets of 3 edge-disjoint
spanning trees. 

\medskip\noindent{\bf Future work and open problems.}
The discussion in Section \ref{sec:genericity} highlights
a practical, but notoriously difficult, problem in rigidity theory: determine geometric necessary or
sufficient conditions for a specific embedding of a body-and-cad framework to be generic.

We suspect that analyzing point-point coincidences in body-and-cad frameworks may be just as hard as analyzing 3D 
bar-and-joint rigidity.  Understanding precisely where the challenge lies merits a comprehensive investigation, and the relationship between 
2D body-and-cad and 2D bar-and-joint structures is a natural starting point.

\medskip\noindent{\bf Algorithms.}
The combinatorial property of Theorem \ref{theorem:minRigIFFTrees} must be matroidal 
as it characterizes the independence of rows in a matrix. 
Due to the intimate relation between edge-disjoint spanning trees and sparsity counts,
we expect to obtain efficient algorithms for body-and-cad rigidity by generalizing the
pebble game algorithms. Such algorithms would not only {\bf decide} whether a given framework is rigid or flexible,
but should be able to {\bf detect rigid components and circuits}, minimally dependent
sets of (primitive) constraints. This would provide valuable feedback to CAD users when the addition 
of a constraint causes an inconsistency.

\bibliography{bodyCad.bib}

\newpage
\appendix
\section{Alternative development of body-and-cad rigidity matrix}
	\label{apx:fix2rows}
	In \cite{haller:etAl:bodyCad:CGTA:2012}, a body-and-cad rigidity matrix was developed by expressing 
	each of 20 possible 3-dimensional constraints in terms of 
	four basic constraints:	\\
	\noindent {\em (i) basic line-line non-parallel fixed angular}, \\ 
	\noindent {\em (ii) basic line-line parallel}, \\
	\noindent {\em (iii) basic blind orthogonality}, and \\
	\noindent {\em (iv) basic blind parallel}.\\
	Point-point coincidence constraints required a separate development, and the resulting rows do not have the 
	pattern of generically non-zero entries required for the $(k,g)$-frames in this paper.
	
	Constraints $(i)$ and $(ii)$ are {\em angular} constraints, as the corresponding twists have $\bv = 0.$
	Each basic constraint corresponds to either one or two rows in the rigidity matrix.
	
	The {\bf basic angular} constraint $(i)$ requires that the angle between two vectors $\ba$ and
	$\bb$ be maintained.   This constraint corresponds to a single row in the rigidity matrix of the form
	\begin{center}
	\begin{tabular}{ccccccc}
	$\cdots$ & $\bv_i$& $-\bo_i$ &
	$\cdots$ & $\bv_j$& $-\bo_j$ & $\cdots$ \\
	\hline
	\multicolumn{1}{|c|}{$\rmfill$} & 
		\multicolumn{1}{c|}{\cellcolor{red}${\bf 0}$} & 
		\multicolumn{1}{c|}{\cellcolor{lightgray}$\bb \times \ba$}&
		\multicolumn{1}{c|}{$\rmfill$} &
		\multicolumn{1}{c|}{\cellcolor{red}${\bf 0}$} & 
		\multicolumn{1}{c|}{\cellcolor{lightgray}$\ba \times \bb$}&	
		\multicolumn{1}{c|}{$\rmfill$} \\
	\hline.
	\end{tabular}
	\end{center}
	Constraint $(ii)$ requires a line to be parallel to a line in a fixed direction $\bc = (c_1, c_2, c_3).$  This condition translates into two rows:	\begin{center}
	\begin{tabular}{ccccccc}
	$\cdots$ & $\bv_i$& $-\bo_i$ &
	$\cdots$ & $\bv_j$& $-\bo_j$ & $\cdots$ \\
	\hline
	\multicolumn{1}{|c|}{$\rmfill$} & 
	\multicolumn{1}{c|}{\cellcolor{red}${\bf 0}$} & 
	\multicolumn{1}{c|}{\cellcolor{lightgray}$(-c_2, c_1, 0)$}&
		\multicolumn{1}{c|}{$\rmfill$} &
		\multicolumn{1}{c|}{\cellcolor{red}${\bf 0}$} & 
		\multicolumn{1}{c|}{\cellcolor{lightgray}$(c_2, -c_1, 0)$}&	
		\multicolumn{1}{c|}{$\rmfill$} \\
	\hline
	\multicolumn{1}{|c|}{$\rmfill$} & 
	\multicolumn{1}{c|}{\cellcolor{red}${\bf 0}$} & 
		\multicolumn{1}{c|}{\cellcolor{lightgray}$(0, -c_3, c_2)$}&
		\multicolumn{1}{c|}{$\rmfill$} &
		\multicolumn{1}{c|}{\cellcolor{red}${\bf 0}$} & 
		\multicolumn{1}{c|}{\cellcolor{lightgray}$(0, c_3, -c_2)$}&	
		\multicolumn{1}{c|}{$\rmfill$} \\
	\hline
	\end{tabular}
	\end{center} 

	The reverse direction of the proof of Theorem~\ref{theorem:minRigIFFTrees} requires us to assume that,	
	except for the vectors $\bv$ that are required to be zero in angular constraints, all other entries of
	the rigidity matrix are generically non-zero.    Therefore, this representation of constraint $(ii)$ is
	incompatible with the proof.
	
	However, since $(ii)$ is expressing a line-line parallel constraint, we have an alternate description.
	Let $\ba$ and $\bb$ be two directions such that
	$\ba \times \bb = \bc$. Then constraint $(ii)$ can be expressed using two rows, each of which
	reduce to basic constraint $(i)$:	
	\begin{center}
	\begin{tabular}{ccccccc}
	$\cdots$ & $\bv_i$& $-\bo_i$ &
	$\cdots$ & $\bv_j$& $-\bo_j$ & $\cdots$ \\
	\hline
	\multicolumn{1}{|c|}{$\rmfill$} & 
		\multicolumn{1}{c|}{\cellcolor{red}${\bf 0}$} & 
		\multicolumn{1}{c|}{\cellcolor{lightgray}$\bc \times \ba$}&
		\multicolumn{1}{c|}{$\rmfill$} &
		\multicolumn{1}{c|}{\cellcolor{red}${\bf 0}$} & 
		\multicolumn{1}{c|}{\cellcolor{lightgray}$\ba \times \bc$}&	
		\multicolumn{1}{c|}{$\rmfill$} \\
	\hline
		\multicolumn{1}{|c|}{$\rmfill$} & 
			\multicolumn{1}{c|}{\cellcolor{red}${\bf 0}$} & 
			\multicolumn{1}{c|}{\cellcolor{lightgray}$\bc \times \bb$}&
			\multicolumn{1}{c|}{$\rmfill$} &
			\multicolumn{1}{c|}{\cellcolor{red}${\bf 0}$} & 
			\multicolumn{1}{c|}{\cellcolor{lightgray}$\bb\times \bc$}&	
			\multicolumn{1}{c|}{$\rmfill$} \\
		\hline
	\end{tabular}
	\end{center}

	For the {\bf basic blind} constraints, let
	$\bp$ be a point, $\bp'$ its instantaneous velocity resulting from a twist and $\bc = (c_1,c_2,c_3)$
	a direction vector. Then $(iii)$ expressed that $\bp'$ must be orthogonal to $\bc$,
	and was associated to a single row 	
	\begin{center}
	\begin{tabular}{ccccccc}
	$\cdots$ & $\bv_i$& $-\bo_i$ &
	$\cdots$ & $\bv_j$& $-\bo_j$ & $\cdots$ \\
	\hline
	\multicolumn{1}{|c|}{$\rmfill$} & 
		\multicolumn{2}{c|}{\cellcolor{lightgray}$(\bp : 1) \vee (\bc : 0)$}&
		\multicolumn{1}{c|}{$\rmfill$} &
		\multicolumn{2}{c|}{\cellcolor{lightgray}$-(\bp : 1) \vee (\bc : 0)$}&	
		\multicolumn{1}{c|}{$\rmfill$} \\
	\hline.
	\end{tabular}
	\end{center}
	Constraint $(iv)$ expressed that $\bp'$ must be parallel to $\bc$
	and was associated to two rows
	\begin{center}
	\begin{tabular}{ccccccc}
	$\cdots$ & $\bv_i$& $-\bo_i$ &
	$\cdots$ & $\bv_j$& $-\bo_j$ & $\cdots$ \\
	\hline
	\multicolumn{1}{|c|}{$\rmfill$} & 
		\multicolumn{2}{c|}{\cellcolor{lightgray}$(\bp : 1) \vee (c_2,-c_1,0,0)$}&
		\multicolumn{1}{c|}{$\rmfill$} &
		\multicolumn{2}{c|}{\cellcolor{lightgray}$-(\bp : 1) \vee (c_2,-c_1,0,0)$}&	
		\multicolumn{1}{c|}{$\rmfill$} \\
	\hline
		\multicolumn{1}{|c|}{$\rmfill$} & 
			\multicolumn{2}{c|}{\cellcolor{lightgray}$(\bp : 1) \vee (0,c_3,-c_2,0)$}&
			\multicolumn{1}{c|}{$\rmfill$} &
			\multicolumn{2}{c|}{\cellcolor{lightgray}$-(\bp : 1) \vee (0,c_3,-c_2,0)$}&	
			\multicolumn{1}{c|}{$\rmfill$} \\
		\hline
	\end{tabular}
	\end{center}

	As in constraint $(ii)$, in constraint $(iv)$ we are requiring that two vectors,
	$\bp'$ and $\bc$ be parallel. Again, there is an alternate description.  Indeed, let $\ba$ and $\bb$ be two directions such that
	$\ba \times \bb = \bc$. Then constraint $(iv)$ can be expressed using two rows, each of which
	reduces to basic constraint $(iii)$:
	\begin{center}
	\begin{tabular}{ccccccc}
	$\cdots$ & $\bv_i$& $-\bo_i$ &
	$\cdots$ & $\bv_j$& $-\bo_j$ & $\cdots$ \\
	\hline
	\multicolumn{1}{|c|}{$\rmfill$} & 
		\multicolumn{2}{c|}{\cellcolor{lightgray}$(\bp : 1) \vee (\ba : 0)$}&
		\multicolumn{1}{c|}{$\rmfill$} &
		\multicolumn{2}{c|}{\cellcolor{lightgray}$-(\bp : 1) \vee (\ba : 0)$}&	
		\multicolumn{1}{c|}{$\rmfill$} \\
	\hline
		\multicolumn{1}{|c|}{$\rmfill$} & 
			\multicolumn{2}{c|}{\cellcolor{lightgray}$(\bp : 1) \vee (\bb : 0)$}&
			\multicolumn{1}{c|}{$\rmfill$} &
			\multicolumn{2}{c|}{\cellcolor{lightgray}$-(\bp : 1) \vee (\bb : 0)$}&	
			\multicolumn{1}{c|}{$\rmfill$} \\
		\hline
	\end{tabular}
	\end{center}
	
	We conclude that, generically, the entries in the three $\bo$ columns for {\bf basic angular} 
	constraints $(i)$ and $(ii)$ are non-zero, and the entries in the six $\bv$ and $\bo$ columns for {\bf basic blind}
	constraints $(iii)$ and $(iv) $are non-zero. Therefore, the body-and-cad rigidity matrix resulting
	from this alternative development will satisfy the definition of a $(6,3)$-frame (as per Definition \ref{def:kgframe}).
	
\section{Point-point coincidence constraints}
\label{apx:ptPtCoinc}
In this section, we provide some insight as to why we are unable to treat point-point coincidence constraints.  
If $\bp = (p^x, p^y, p^z)$ is the point of coincidence, then the constraint is infinitesimally expressed by 3 
rows in the rigidity matrix of the form:
\begin{center}
\begin{tabular}{ccccccc}
\hline
\multicolumn{1}{|c|}{$\rmfill$} & 
	\multicolumn{2}{c|}{\cellcolor{lightgray}$(1, 0, 0, 0, -p^z, p^y)$}&
	\multicolumn{1}{c|}{$\rmfill$} &
	\multicolumn{2}{c|}{\cellcolor{lightgray}$(-1, 0, 0, 0, p^z, -p^y)$}&	
	\multicolumn{1}{c|}{$\rmfill$} \\
\hline
\multicolumn{1}{|c|}{$\rmfill$} & 
	\multicolumn{2}{c|}{\cellcolor{lightgray}$(0, 1, 0, p^z, 0, -p^x)$}&
	\multicolumn{1}{c|}{$\rmfill$} &
	\multicolumn{2}{c|}{\cellcolor{lightgray}$(0, -1, 0, -p^z, 0, p^x)$}&	
	\multicolumn{1}{c|}{$\rmfill$} \\
\hline
\multicolumn{1}{|c|}{$\rmfill$} & 
	\multicolumn{2}{c|}{\cellcolor{lightgray}$(0, 0, 1, -p^y, p^x, 0)$}&
	\multicolumn{1}{c|}{$\rmfill$} &
	\multicolumn{2}{c|}{\cellcolor{lightgray}$(0, 0, -1, p^y, -p^x, 0)$}&	
	\multicolumn{1}{c|}{$\rmfill$} \\
\hline
\end{tabular}
\end{center}
Since these rows do not conform to the pattern required by our proof technique,
point-point coincidences are not addressed by the characterization.

\section{Body-and-cad algebraic rigidity theory}
\label{apx:algebraicRT}
For completeness, we include the definitions from \cite{haller:etAl:bodyCad:CGTA:2012} required to express the geometry of
a body-and-cad structure, from which the algebraic theory is derived.
A \emph{cad graph} $(G,c)$ is used to represent the combinatorics of the
constraints,  where
$G = (V,E)$ is a multigraph with $V = \{1,\ldots,n\}$ and
$c: E \rightarrow C$ is an edge coloring using colors $C = \{c_1, \ldots, c_{21}\}$
to denote the 21 possible cad constraints.
This is related to the {\em primitive cad graph} presented in Section \ref{sec:primitive}:
for each color $c_i$, we associate some number of primitive angular and primitive blind
constraints as found in Table \ref{table:constraintAssoc} (reproduced from \cite{haller:etAl:bodyCad:CGTA:2012}).
The bi-colored primitive cad graph $H = (V,E=B \disjointUnion R)$ has a red edge for each
primitive angular constraint and a black edge for each primitive blind constraint.

\begin{table}[bht]
\begin{center}
	\begin{tabular}{|l||c|c|c|c|c|c|}
\hline  & \multicolumn{2}{c|}{\bf point} & \multicolumn{2}{c|}{\bf
line} &
 \multicolumn{2}{c|}{\bf plane} \\
\hline
 & angular & blind & angular & blind & angular & blind \\
\hline  \multicolumn{7}{|l|}{\bf point} \\
\hline  coincidence & 0 & 3 & 0 & 2 & 0 & 1 \\
 distance & 0 & 1 & 0 & 1 & 0 & 1 \\
\hline \multicolumn{7}{|l|}{\bf line} \\
\hline 
 coincidence &  &  & 2 & 2 & 1 & 1 \\
 distance &  &  & 0 & 1 & 1 & 1 \\
 parallel &  &  & 2 & 0 & 1 & 0 \\
 perpendicular &  &  & 1 & 0 & 2 & 0 \\
 fixed angular &  &  & 1 & 0 & 1 & 0 \\
\hline  \multicolumn{7}{|l|}{\bf plane} \\
\hline 
 coincidence & & & & & 2 & 1  \\
 distance & & & & & 2 & 1  \\
parallel & & & & & 2 & 0 \\
 perpendicular  & & & & & 1 & 0 \\
 fixed angular & & & & & 1 & 0\\
 \hline
\end{tabular}
\end{center}
\caption{Association of body-and-cad ({\em coincidence, angular, distance}) constraints 
with the number of blind and angular primitive constraints. As an example of how to read
the table, the 
last two columns (corresponding to {\bf plane}) of row 3 (corresponding to {\bf coincidence}    
under {\bf line}) indicate
that a {\bf line-plane
coincidence} constraint reduces to 1 angular 
and 1 blind primitive constraint.} 
\label{table:constraintAssoc}
\end{table}

The geometry of the constraint represented by each cad graph edge is
given in terms of points, lines, or planes affixed to the $n$ rigid bodies; each line or
plane may be described by a point on a rigid body and a direction vector.  Let $E_i$ denote the set of edges of color $c_i.$  
Then we define a ``length'' function $L_i$ that associates to each edge in $E_i$ the real-valued points, vectors and values 
(e.g., specified distance or angle) needed to describe the constraint; the full description of the family of functions $L_i$ 
is below. With these definitions, we can formally define a \emph{body-and-cad framework} to be $(G,c, \bl)= (G,c, L_1, \ldots, L_{21})$, 
the cad graph $(G,c)$ together with the functions $L_i$.

\begin{itemize}
	\item {\bf Point-point coincidence:}
	$L_1:E_1 \rightarrow \RR^3$ maps an edge $e=ij$ to a point $\vec p$ so that it is constrained
	to lie on bodies $i$ and $j$ simultaneously.

	\item {\bf Point-point distance:}
	$L_2:E_2 \rightarrow \RR^3 \times \RR^3 \times \RR$ maps an edge $e=ij$ to a triple $(\vec p_i, \vec p_j, a)$ 
	so that the point $\vec p_i$ affixed to body $i$ is constrained to lie a distance $a$ from the 
	point $\vec p_j$ affixed to body $j$.
	
	\item {\bf Point-line coincidence:}
	$L_3:E_3 \rightarrow \RR^3 \times (\RR^3 \times \RR^3)$
	maps an edge $e=ij$ to a pair $(\vec p_i, (\vec p_j, \vec d))$ so that
	point $\vec p_i$ affixed to body $i$
	is constrained to lie on the line $(\vec p_j, \vec d)$ affixed to body $j$.
	
	\item {\bf Point-line distance:}
	$L_4:E_4 \rightarrow \RR^3 \times (\RR^3 \times \RR^3) \times \RR$
	maps an edge $e=ij$ to a triple $(\vec p_i, (\vec p_j, \vec d), a)$ so that
	point $\vec p_i$ affixed to body $i$
	is constrained to lie a distance $a$ from the line $(\vec p_j, \vec d)$ affixed to body $j$.
	
	\item {\bf Point-plane coincidence:} 
	$L_5:E_5 \rightarrow \RR^3 \times (\RR^3 \times \RR^3)$
	maps an edge $e=ij$ to a pair $(\vec p_i, (\vec p_j, \vec d))$ so that
	the point $\vec p_i$ affixed to body $i$ is constrained to lie
	in the plane $(\vec p_j, \vec d)$ affixed to body $j$.
	
	\item {\bf Point-plane distance:} 
	$L_6:E_6 \rightarrow \RR^3 \times (\RR^3 \times \RR^3) \times \RR$
	maps an edge $e=ij$ to a triple $(\vec p_i, (\vec p_j, \vec d), a)$ so that
	the point $\vec p_i$ affixed to body $i$ is constrained to lie a distance $a$
	from the plane $(\vec p_j, \vec d)$ affixed to body $j$.
	
	\item {\bf Line-line parallel:} 
	$L_{7}:E_{7} \rightarrow \RR^3 \times \RR^3 \times \RR^3$
	maps an edge $e=ij$ to a triple $(\vec p_i, \vec p_j, \vec d)$ so that
	the lines $(\vec p_i, \vec d)$ and $(\vec p_j, \vec d)$ affixed to bodies $i$ and $j$, respectively,
	are constrained to remain parallel to each other.
	
	\item {\bf Line-line perpendicular:} 
	$L_{8}:E_{8} \rightarrow (\RR^3 \times \RR^3) \times (\RR^3 \times \RR^3)$
	maps an edge $e=ij$ to a pair $((\vec p_i, \vec d_i), (\vec p_j, \vec d_j))$ so that
	the lines $(\vec p_i, \vec d_i)$ and $(\vec p_j, \vec d_j)$ affixed to bodies $i$ and $j$, respectively,
	are constrained to remain perpendicular to each other.
	
	\item {\bf Line-line fixed angular:} 
	$L_{9}:E_{9} \rightarrow (\RR^3 \times \RR^3) \times (\RR^3 \times \RR^3) \times \RR$
	maps an edge $e=ij$ to a triple $((\vec p_i, \vec d_i), (\vec p_j, \vec d_j), \alpha)$ so that
	the lines $(\vec p_i, \vec d_i)$ and $(\vec p_j, \vec d_j)$ affixed to bodies $i$ and $j$, respectively,
	are constrained to maintain the angle $\alpha$ between them.
	
	\item {\bf Line-line coincidence:}
	$L_{10}:E_{10} \rightarrow \RR^3 \times \RR^3$
	maps an edge $e=ij$ to a pair $(\vec p, \vec d)$ so that
	the line $(\vec p, \vec d)$ is constrained to be affixed to bodies $i$ and $j$ 
	simultaneously.
	
	\item {\bf Line-line distance:}
	$L_{11}:E_{11} \rightarrow (\RR^3 \times \RR^3) \times (\RR^3 \times \RR^3) \times \RR$
	maps an edge $e=ij$ to a triple $((\vec p_i, \vec d_i), (\vec p_j, \vec d_j), a)$ so that
	the lines $(\vec p_i, \vec d_i)$ and $(\vec p_j, \vec d_j)$ affixed to bodies $i$ and $j$, 
	respectively, are constrained to lie a distance $a$ from each other.
	
	\item {\bf Line-plane parallel:} 
	$L_{12}:E_{12} \rightarrow (\RR^3 \times \RR^3) \times (\RR^3 \times \RR^3)$
	maps an edge $e=ij$ to a pair $((\vec p_i, \vec d_i), (\vec p_j, \vec d_j))$ so that
	the line $(\vec p_i, \vec d_i)$ and plane $(\vec p_j, \vec d_j)$ affixed to bodies $i$ and $j$, respectively,
	are constrained to remain parallel to each other.
	
	\item {\bf Line-plane perpendicular:} 
	$L_{13}:E_{13} \rightarrow (\RR^3 \times \RR^3) \times (\RR^3 \times \RR^3)$
	maps an edge $e=ij$ to a pair $((\vec p_i, \vec d_i), (\vec p_j, \vec d_j))$ so that
	the line $(\vec p_i, \vec d_i)$ and plane $(\vec p_j, \vec d_j)$ affixed to bodies $i$ and $j$, respectively,
	are constrained to remain perpendicular to each other.
	
	\item {\bf Line-plane fixed angular:} 
	$L_{14}:E_{14} \rightarrow (\RR^3 \times \RR^3) \times (\RR^3 \times \RR^3) \times \RR$
	maps an edge $e=ij$ to a triple $((\vec p_i, \vec d_i), (\vec p_j, \vec d_j), \alpha)$ so that
	the line $(\vec p_i, \vec d_i)$ and plane $(\vec p_j, \vec d_j)$ affixed to bodies $i$ and $j$, respectively,
	are constrained to maintain the angle $\alpha$ between them.

	\item {\bf Line-plane coincidence:}
	$L_{15}:E_{15} \rightarrow (\RR^3 \times \RR^3) \times (\RR^3 \times \RR^3)$
	maps an edge $e=ij$ to a pair $((\vec p_i, \vec d_i), (\vec p_j, \vec d_j))$ so that
	the line $(\vec p_i, \vec d_i)$ affixed to body $i$ is constrained to lie
	in the plane $(\vec p_j, \vec d_j)$ affixed to body $j$.
	
	\item {\bf Line-plane distance:}
	$L_{16}:E_{16} \rightarrow (\RR^3 \times \RR^3) \times (\RR^3 \times \RR^3) \times \RR$
	maps an edge $e=ij$ to a triple $((\vec p_i, \vec d_i), (\vec p_j, \vec d_j), a)$ so that
	the line $(\vec p_i, \vec d_i)$ affixed to body $i$ is constrained to lie a distance $a$
	from the plane $(\vec p_j, \vec d_j)$ affixed to body $j$.
	
	\item {\bf Plane-plane parallel:} 
	$L_{17}:E_{17} \rightarrow \RR^3 \times \RR^3 \times \RR^3$
	maps an edge $e=ij$ to a triple $(\vec p_i, \vec p_j, \vec d)$ so that
	the planes $(\vec p_i, \vec d)$ and $(\vec p_j, \vec d)$ affixed to bodies $i$ and $j$, respectively,
	are constrained to remain parallel to each other.
	
	\item {\bf Plane-plane perpendicular:} 
	$L_{18}:E_{18} \rightarrow (\RR^3 \times \RR^3) \times (\RR^3 \times \RR^3)$
	maps an edge $e=ij$ to a pair $((\vec p_i, \vec d_i), (\vec p_j, \vec d_j))$ so that
	the planes $(\vec p_i, \vec d_i)$ and $(\vec p_j, \vec d_j)$ affixed to bodies $i$ and $j$, respectively,
	are constrained to remain perpendicular to each other.
	
	\item {\bf Plane-plane fixed angular:} 
	$L_{19}:E_{19} \rightarrow (\RR^3 \times \RR^3) \times (\RR^3 \times \RR^3) \times \RR$
	maps an edge $e=ij$ to a triple $((\vec p_i, \vec d_i), (\vec p_j, \vec d_j), \alpha)$ so that
	the planes $(\vec p_i, \vec d_i)$ and $(\vec p_j, \vec d_j)$ affixed to bodies $i$ and $j$, respectively,
	are constrained to maintain the angle $\alpha$ between them.
	
	\item {\bf Plane-plane coincidence:} 
	$L_{20}:E_{20} \rightarrow \RR^3 \times \RR^3 $
	maps an edge $e=ij$ to the pair $(\vec p, \vec d)$ so that
	the plane $(\vec p, \vec d)$ 
	is constrained to be affixed to both bodies $i$ and $j$ simultaneously.

	\item {\bf Plane-plane distance:}
	$L_{21}:E_{21} \rightarrow \RR^3 \times \RR^3 \times \RR^3 \times \RR$
	maps an edge $e=ij$ to a quadruple $(\vec p_i, \vec p_j, \vec d, a)$ so that
	the planes $(\vec p_i, \vec d)$ and $(\vec p_j, \vec d)$ affixed to bodies $i$ and $j$, respectively,
	are constrained to have the distance $a$ between them. 
\end{itemize}
\end{document}